\documentclass{article}

\usepackage{microtype}

\usepackage{booktabs} 
\usepackage{amsmath}
\usepackage{amssymb}
\usepackage{amsthm}
\usepackage{amsfonts}
\usepackage{xspace}
\usepackage{algorithm}
\usepackage{algorithmicx}
\usepackage{algpseudocode}
\usepackage{color}
\usepackage[T1]{fontenc}
\usepackage{bm}
\usepackage{thmtools}
\usepackage{lmodern}

\newcommand{\net}{\ensuremath{\mathfrak{N}}\xspace}
\newcommand{\nat}{\ensuremath{\mathbb{N}}\xspace}
\newcommand{\Prob}[1]{\mathbb{P}\left[#1\right]}
\newcommand{\ci}{\ensuremath{c}\xspace}
\newcommand{\cj}{\ensuremath{d}\xspace}
\newcommand{\ecc}{\ensuremath{D}\xspace}

\title{Deterministic Blind Radio Networks }

\author{Artur Czumaj \and Peter Davies}

\newtheorem{theorem}{Theorem}

\newtheorem{lemma}[theorem]{Lemma}

\newtheorem{definition}[theorem]{Definition}

\begin{document}

\maketitle

\begin{abstract}
Ad-hoc radio networks and multiple access channels are classical and well-studied models of distributed systems, with a large body of literature on deterministic algorithms for fundamental communications primitives such as broadcasting and wake-up. However, almost all of these algorithms assume knowledge of the number of participating nodes and the range of possible IDs, and often make the further assumption that the latter is linear in the former. These are very strong assumptions for models which were designed to capture networks of weak devices organized in an ad-hoc manner. It was believed that without this knowledge, deterministic algorithms must necessarily be much less efficient.

In this paper we address this fundamental question and show that this is not the case. We present \emph{deterministic} algorithms for \emph{blind} networks (in which nodes know only their own IDs), which match or nearly match the running times of the fastest algorithms which assume network knowledge (and even surpass the previous fastest algorithms which assume parameter knowledge but not small labels).

Specifically, in multiple access channels with $k$ participating nodes and IDs up to $L$, we give a wake-up algorithm requiring $O(\frac{k\log L \log k }{\log\log k})$ time, improving dramatically over the $O(L^3 \log^3 L)$ time algorithm of De Marco et al. (2007), and a broadcasting algorithm requiring \sloppy{$O(k\log L \log\log k)$ }time, improving over the $O(L)$ time algorithm of G\k{a}sieniec et al. (2001) in most circumstances. Furthermore, we show how these same algorithms apply directly to multi-hop radio networks, achieving even larger running time improvements.
 \end{abstract}

\section{Introduction}

In this paper we address the fundamental question in distributed computing of whether basic communication primitives can be efficiently performed in networks in which the participating nodes have no knowledge about the network structure. Our focus is on \emph{deterministic} algorithms.

\subsection{Models and problems}

We consider the two classical, and related, models of distributed communication: \emph{multiple access channels} (cf. \cite{-JS05,-W86}) and \emph{ad-hoc multi-hop radio networks} (cf. \cite{-BGI92,-CGK07,-CR06,-Pel07}).

\paragraph{Multiple access channels.}

A set of $k$ nodes, with unique identifiers (IDs) from $\{1, \dots, L\}$, share a communication channel. Time is divided into discrete steps, and in every step each node chooses to either transmit a message to the channel or listen for messages. A transmission is only successful if exactly one node chooses to transmit in a given time-step; otherwise all nodes hear silence.

\paragraph{Ad-hoc multi-hop radio networks.}

The network is modeled by a \emph{directed} graph $\net = (V,E)$, with $|V| = n$, where nodes correspond to transmitter-receiver stations. The nodes have unique identifiers from $\{1, \dots, L\}$. A directed edge $(v,u) \in E$ means that node $v$ can send a message directly to node $u$. To make propagation of information feasible, we assume that every node in $V$ is reachable in \net from any other. Time is divided into discrete steps, and in every step each node chooses to either transmit a message to all neighbors or listen for messages. A listening node only hears a transmission if exactly one neighbor transmitted; otherwise it hears silence.

It can be seen that multiple access channels are equivalent to \emph{single-hop radio networks} (that is, radio networks in which the underlying graph is a clique).

\paragraph{Node knowledge.}

We study \emph{blind} versions of these models, by which we mean that the minimum possible assumptions about node knowledge are made (and this is where our work differs most significantly from previous work): we assume nodes do not know the parameters $k$, $L$, or $n$, or any upper bounds thereof. In accordance with the standard model of ad-hoc radio networks (for more elaborate discussion about the model, see, e.g., \cite{-ABLP91,-BGI92,-CGGPR00,-CGR00,-CMS03,-GHK13,-KP03,-KM98,-Pel07}), we also make the assumption that a node does not have any  prior knowledge about the topology of the network, its in-degree and out-degree, or the set of its neighbors. In our setting, \emph{the only prior knowledge nodes have is their own unique ID}.

\paragraph{Tasks.}

In both models we consider the fundamental communication tasks of \emph{broadcasting} (see, e.g., the survey \cite{-Pel07} and the references therein) and \emph{wake-up} (cf. \cite{-CGKR05,-CGK07,-GPP01}).

In the task of \emph{wake-up}, nodes begin in a dormant state, and some non-empty subset of nodes spontaneously `wake up' at arbitrary (adversarially chosen) time-steps. Nodes are also woken up if they receive messages. Nodes cannot participate (by transmitting) until they are woken up, and our goal is to ensure that eventually all nodes are awake. We assume nodes have access only to a \emph{local clock}: they can count the number of time-steps since they woke up, but there is no global awareness of an absolute time-step number.

The task of \emph{broadcasting} is usually described as follows: one node begins with a message, and it must inform all other nodes of this message via transmissions. However, to enable our results to transfer from multiple access channels (single-hop radio networks) to multi-hop radio networks, we will instead use broadcasting to refer to a more generalized task. Our broadcasting task will be defined similarly to wake-up, with the only difference being that nodes have access to a \emph{global clock}, informing them of the absolute time-step number. (In multiple access channels, this task is usually also referred to as wake-up, specifying global clock access, but here we will call it broadcasting to better differentiate.)

Notice that the standard broadcasting task in radio networks is a special case of this task, in which only one node spontaneously wakes up. A global clock can be simulated by appending the current global time-step to each transmitted message (and since all message chains originate from the same source node, these time-steps will agree).

For both tasks, we wish to minimize the number of time-steps that elapse between the first node waking up, and all nodes being woken. We are not concerned with the computation performed by nodes within time-steps.

\subsection{Related work}

As fundamental communications primitives, the tasks of designing efficient deterministic algorithms for \emph{broadcasting} and \emph{wake-up} have been extensively studied for various network models for many decades.

\paragraph{Wake-up.}

The wake-up problem (with only local clocks) has been studied in both multiple access channels and multi-hop radio networks (often separately, though the results usually transfer directly from one to the other). It has been commonly assumed in the literature that network parameters are known, and that IDs are small ($L = n^{O(1)}$).

The first sub-quadratic deterministic wake-up protocol for radio networks was given in by Chrobak et al.\ \cite{-CGK07}, who introduced the concept of \emph{radio synchronizers} to abstract the essence of the problem. They give an $O(n^{5/3} \log n)$-time protocol for the wake-up problem. Since then, there have been several improvements in running time, making use of the radio synchronizer machinery: firstly to $O(n^{3/2} \log n)$ \cite{-CK04}, and then to $O(n \log^2 n)$ \cite{-CGKR05}. The current fastest protocol is $O(\frac{n \log ^2 n }{\log\log n})$ \cite{-CD16}. However, without the assumption of small labels, all of these running times are increased. The algorithm of \cite{-CD16} as analyzed would give $O(\frac{n \log L \log (n\log L) }{\log\log (n \log L)})$ time, and similar time with $k$ replacing $n$ in multiple access channels. All of these algorithms, like those we present here, are non-explicit.

There has been some work on wake-up in multiple access channels without knowledge of network parameters: firstly an $O(L^4 \log^5 L)$ algorithm \cite{-GPP01}, and then an improvement to $O(L^3 \log^3 L)$ \cite{-DMPS07}. It was believed that this algorithms in this setting were necessarily much slower than those for when parameters were known; for example, \cite{-DMPS07} states ``a crucial assumption is whether the processors using the shared channel are aware of the total number $n$ of processors sharing the channel, or some polynomially related upper bound to such number. When such number $n$ is known, much faster algorithms are possible.''

There are no direct results for wake-up in radio networks with unknown parameters, but the algorithm of \cite{-DMPS07} can be applied to give $O(n L^3 \log^3 L)$ time.

We note that randomized algorithms for wake-up have also been studied, both with and without parameter knowledge; see \cite{-GPP01,-JS05}.

\paragraph{Broadcasting.}

Broadcasting is possibly the most studied problem in radio networks, and has a wealth of literature in various settings. For the model studied in this paper, \emph{directed} radio networks with \emph{unknown structure} and \emph{without collision detection}, the first sub-quadratic \emph{deterministic} broadcasting algorithm was proposed by Chlebus et al.\ \cite{-CGGPR00}, who gave an $O(n^{11/6})$-time broadcasting algorithm. 
After several small improvements (cf. \cite{-CGOR00,-DMP01}), Chrobak et al.\ \cite{-CGR00} designed an almost optimal algorithm that ns the task in $O(n \log^2n)$ time, the first to be only a poly-logarithmic factor away from linear dependency. Kowalski and Pelc \cite{-KP03} improved this bound to obtain an algorithm of complexity $O(n \log n \log\ecc)$ and Czumaj and Rytter \cite{-CR06} gave a broadcasting algorithm running in time $O(n \log^2\ecc)$. Here $\ecc$ is the eccentricity of the network, i.e., the distance between the furthest pair of nodes. De Marco \cite{-DM10} designed an algorithm that completes broadcasting in $O(n \log n \log\log n)$ time steps, beating \cite{-CR06} for general graphs. Finally, the $O(n \log D \log\log D)$ algorithm of \cite{-CD16} came within a log-logarithmic factor of the $\Omega(n \log D)$ lower bound \cite{-CMS03}. Again, however, these results generally assume \emph{small node labels} ($L=O(n)$, though some of the earlier results only require $L=O(n^c)$ for some constant $c$), and their running time results do not hold otherwise. The situation where node labels can be large is less well-studied, though it is easy to see that the algorithm of \cite{-CGR00} still works, requiring $O(n\log^2 L)$ time. In multiple access channels, a $O(k\log\frac Lk)$ time algorithm exists \cite{-CMS03}. Again, all of these algorithms are, like those presented here, non-explicit.

All of these results also \emph{intrinsically require parameter knowledge}. Without knowledge of $n$, $L$, $k$, or $D$, the fastest algorithm known is the $O(L)$ time algorithm of \cite{-GPP01} for multiple access channels. This algorithm is explicit, but has the strong added restriction that the first node wakes up at global time-step $0$. It also does not transfer to multi-hop radio networks, so the best running time therein is the $O(D L^3 \log^3 L)$ given by the algorithm of \cite{-DMPS07}. Concurrently with this work, randomized algorithms for broadcasting without parameter knowledge are presented in \cite{-CD18b}, achieving a nearly optimal running time of $O(D\log\frac nD\log^2\log \frac nD + \log^2 n)$ in the model we study here (that without collision detection).

Broadcasting, as a fundamental communication primitive, has been also studied in various related models, including undirected networks, randomized broadcasting protocols, models with collision detection, and models in which the entire network structure is known. For example, if the underlying network is undirected, then an $O(n \log \ecc)$-time algorithm due to Kowalski \cite{-K05} exists. If spontaneous transmissions are allowed and a global clock available, then deterministic broadcast can be performed in $O(L)$ time in undirected networks \cite{-CGGPR00}. Randomized broadcasting has been also extensively studied, and in a seminal paper, Bar-Yehuda et al.\ \cite{-BGI92} designed an almost optimal broadcasting algorithm achieving the running time of $O((\ecc + \log n) \cdot \log n)$. This bound has been later improved by Czumaj and Rytter \cite{-CR06}, and independently Kowalski and Pelc \cite{-KP03b}, who gave optimal randomized broadcasting algorithms that complete the task in $O(\ecc \log \frac{n}{\ecc} + \log^2 n)$ time with high probability, matching a known lower bound from \cite{-KM98}.
Haeupler and Wajc \cite{-HW16} improved this bound for undirected networks in the model that allows spontaneous transmissions and designed an algorithm that completes broadcasting in time $O\left(\frac{\ecc \log n \log \log n}{\log\ecc} + \log^{O(1)}n\right)$ with high probability, improved to $O\left(\frac{\ecc \log n}{\log\ecc} + \log^{O(1)}n\right)$ in \cite{-CD17}. In the model with collision detection for undirected networks, an $O(D+\log^6 n)$-time randomized algorithm due to Ghaffari et al. \cite{-GHK13} is the first to exploit collisions and surpass the algorithms (and lower bound) for broadcasting without collision detection.

For more details about broadcasting algorithms in various models, see e.g., \cite{-CD17,-CR06,-GHK13,-K05,-Pel07} and the references therein.


\subsection{New results}

We present algorithms for the fundamental tasks of broadcasting and wake-up in multiple access channels (single-hop radio networks) and multi-hop radio networks which require no knowledge of network parameters: nodes need know only their own unique ID.

Our wake-up algorithm takes $O(\frac{k \log L \log k}{\log\log k})$ time in multiple access channels and $O(\frac{n \log L \log n}{\log\log n})$ time in multi-hop radio networks, improving dramatically over the previous best $O(L^3 \log^3 L)$ and $O(DL^3 \log^3 L)$ respective running times of \cite{-DMPS07} (recall that $k \le n \le L$). This is particularly significant in the case of large labels, since dependency on $L$ has been improved from cubic to logarithmic. Furthermore, our running time matches the $O(\frac{n \log L \log n}{\log\log n})$ time of \cite{-CD16}, the fastest algorithm with parameter knowledge and small node labels.

Our broadcasting algorithm takes $O(k \log L \log\log k)$ time in multiple access channels and $O(n \log L \log\log n)$ time in multi-hop radio networks. This improves over the previous best $O(L)$ multiple access channel bound \cite{-GPP01} in most cases. In radio networks the improvement is even more pronounced, beating not only the $O(DL^3 \log^3 L)$ result of \cite{-DMPS07} but also the $O(n\log^2 L)$-time algorithm of \cite{-CGR00}, which was the fastest result for large labels even when network parameters are known. When labels are small (i.e., $L = n^{O(1)}$), our result matches the best running time for known parameters ($O(n \log D \log\log D)$ from \cite{-CD16}) for networks of polynomial eccentricity.


We believe the primary value of our work is in challenging the \emph{long-standing assumption that parameter knowledge is necessary for efficient deterministic algorithms} in radio networks and multiple access channels. We show that in fact, deterministic algorithms which do not assume this knowledge can reach the fastest running times for those that do.

\subsection{Previous approaches}

Almost all deterministic broadcasting protocols with sub-quadratic complexity (that is, since \cite{-CGGPR00}) have relied on the concept of \emph{selective families} (or some similar variant thereof, such as selectors). These are families of sets for which one can guarantee that any subset of $[k] := \{1,2,\dots,k\}$ below a certain size has an intersection of size exactly $1$ with some member of the family \cite{-CGGPR00}. They are useful in the context of radio networks because if the members of the family are interpreted to be the set of nodes which are allowed to transmit in a particular time-step, then after going through each member, any node with a participating in-neighbor and an in-neighborhood smaller than the size threshold will be informed. Most of the recent improvements in broadcasting time have been due to a combination of proving smaller selective families exist, and finding more efficient ways to apply them (i.e., choosing which size of family to apply at which time) \cite{-CGGPR00,-CGOR00,-CGR00,-CR06}.

Applying such constructs requires coordination between nodes, which relies on a global clock, making them unsuitable for wake-up. To tackle this issue, Chrobak et al.\ \cite{-CGK07} introduced the concept of \emph{radio synchronizers}. These are a development of selective families which allow nodes to begin their behavior at different times. A further extension to \emph{universal synchronizers} in \cite{-CK04} allowed effectiveness across all in-neighborhood sizes.

Another similar extension of selective families came in 2010 with a paper by De Marco \cite{-DM10}, which used a \emph{transmission matrix} to schedule node transmissions for broadcasting. Like radio synchronizers, the application of this matrix allowed nodes to begin their own transmission sequence at any time, and still provided a `selective' property that guaranteed broadcasting progress. The synchronization afforded by the assumption of a global clock allowed this method to beat the time bounds given by radio synchronizers (and previous broadcasting algorithms using selective families).

The proofs of \emph{existence} for selective families, synchronizers, and transmission matrices follow similar lines: a probabilistic candidate object is generated by deciding on each element independently at random with certain carefully chosen probabilities, and then it is proven that the candidate satisfies the desired properties with positive probability, and so \emph{such an object must exist}. These types of proofs are all \emph{non-constructive} (and therefore all resulting algorithms non-explicit; cf. \cite{-CK05,-I02} for an explicit construction of selective families with significantly weaker size bounds).

In contrast, results on multiple access channels without parameter knowledge (notably \cite{-GPP01,-DMPS07}) have not used these types of combinatorial objects, and instead rely on some results from number theory. The algorithm of \cite{-DMPS07}, for instance, is to have nodes transmit periodically, a node with ID $v$ waiting $p_v$ steps between transmissions, where $p_v$ is the $v^{th}$ smallest prime number. A number-theoretic result is then employed to show that eventually one node will transmit alone. As a result, these algorithms have the advantage of being explicit, but the disadvantage of slower running times.

\subsection{Novel approach}

We aim to apply the approach of using combinatorial objects proven by the probabilistic method to the setting where network parameters are not known. One way to do this would be to apply selectors (for example) of continually increasing size, until one succeeds. However, since there are two parameters which must meet the correct values for a successful application ($k$ and $L$ in the case of medium access channels), running times for this approach are poor. Instead, we define, and prove the existence of, \emph{a single, infinite combinatorial object, which can accommodate all possible values of parameters at the same time}.

Another difference is that for all previous works using selective families or variants thereof, the candidate object has been generated with the same sequence of probabilities for each node. Here, however, in order to achieve good running times we need to have these probabilities depending on the node ID. In essence, this means that nodes effectively use their own ID as an estimate of the maximum ID in the network.

\subsection{A note on non-explicitness}
As mentioned, almost all deterministic broadcasting protocols with sub-quadratic complexity have relied on selective families or variants thereof, and have been non-explicit results. Our work here is also non-explicit, but while this is standard for deterministic radio network algorithms, it may present more of an issue here, since our combinatorial structures are infinite. It is not clear how the protocols we present could be performed by devices with bounded memory, and as such this work is more of a proof-of-concept than a practical algorithm. However, it is possible that our method could be adapted so that nodes' behavior could be generated by a finite-size (i.e. a function of ID) program; this would be a natural and interesting extension to our work, and would overcome the problem.

Another issue which would remain is that nodes must perform the protocol indefinitely, and never become aware that broadcasting has been successfully completed. However, this is unavoidable in the model: Chlebus et al. \cite{-CGGPR00} prove that \emph{acknowledged} broadcasting without parameter knowledge is impossible.

\section{Combinatorial objects}

In this section we present the two combinatorial objects that we wish to use in our algorithms: \emph{unbounded universal synchronizers} and \emph{unbounded transmission schedules}. After defining them in Sections \ref{subsec-UUS} and \ref{subsec-UTS}, we present their main properties in Theorems \ref{thm:UURS} and \ref{thm:UTS}, and then show how to apply them to obtain new deterministic algorithms for wake-up and broadcasting in multiple access channels and in radio networks (Theorems \ref{thm:MACWU}, \ref{thm:MACBC}, \ref{thm:RNWU}, \ref{thm:RNB}).

\subsection{Unbounded universal synchronizers}
\label{subsec-UUS}

For the task of wake-up, i.e., in the absence of a global clock, we will define an object called an \textbf{unbounded universal synchronizer} for use in our algorithm.

We begin by defining the sets of circumstances our algorithm must account for:

\begin{definition}
	An \textbf{$(r,k)$-instance} $X$ is a set $K$ of $k$ nodes with \[\sum_{v\in K} \log v  = r\] and wake-up function $\omega:K\rightarrow \nat$.
\end{definition}

(By using $v$ as an integer here, we are abusing notation to mean the ID of node $v$.)

Here $r$ is the main parameter we will use to quantify how `large' our input instance is. By using the sum of logarithms of IDs (which approximates the total number of bits needed to write all IDs in use), we capture both the number of participating nodes and the length of IDs in a single parameter. We require $r$ to be an integer, so we round down accordingly, but we omit floor functions for clarity since they do not affect the asymptotic result.

The \emph{wake-up function} $\omega$ maps each node to the time-step it wakes up (either spontaneously or by receiving a transmission) when our algorithm is run on this instance. This means that the wake-up function depends on the algorithm, but there is no circular dependency: whether nodes wake-up in time-step $j$ only depends on the algorithm's behavior in previous time-steps, and the algorithm's behavior at time-step $j$ only depends on the wake-up function up to $j$. We will also extend $\omega$ to sets of nodes in the instance by $\omega(K):=\min_{v\in K} \omega(v)$.

We now define the combinatorial object that will be the basis of our algorithm:

\begin{definition}
	\label{def:UURS}
	For a function $g:\nat \times \nat \rightarrow \nat$, an \textbf{unbounded universal synchronizer of delay $g$} is a function $\mathcal S : \nat \rightarrow \{0,1\}^\nat$ such that for any $(r,k)$-instance, there is some time-step $j \leq \omega(K) + g(r,k)$ with $\sum_{v\in K} \mathcal S(v)_{j-\omega(v)} = 1$.
\end{definition}

The \emph{unbounded universal synchronizer} $S$ is a function mapping node IDs to a sequence of $0$s and $1$s, which tell nodes when to listen and transmit respectively. The function $g$, which we will call the \emph{delay function}, tells us how many time-steps we must wait before a successful transmission is guaranteed, so this is what we want to asymptotically minimize.

We will apply this object to perform wake-up as follows: each node $v$ transmits a message in time-step $j$ (with its time-step count starting upon waking up) iff $S(v)_{j} = 1$. Then, the property guarantees that at some time-step $j$ within $g(r,k)$ time-steps of the first node waking up, any $(r,k)$-instance will have a successful transmission. We call this $S$ \emph{`hitting' the $(r,k)$-instance at time-step $j$}. So, our aim is to show the existence of such an object, with $g$ growing as slowly as possible.

Our main technical result in this section is the following:

\begin{restatable}{theorem}{UURS}
	\label{thm:UURS}
	There exists an \textbf{unbounded universal synchronizer of delay $g$}, where $g(r,k) =  O\left(\frac{r \log k}{\log\log k}\right)$.
\end{restatable}

Our approach to proving Theorem \ref{thm:UURS} will be to randomly generate a candidate synchronizer, and then prove that with positive probability it does indeed satisfy the required property. Then, for this to be the case, at least one such object must exist.

Our candidate $S : \nat \rightarrow \{0,1\}^\nat$ will be generated by independently choosing each $S(v)_j$ to be \textbf 1 with probability $\frac{\ci\log v}{j+2\ci\log v}$ and \textbf 0 otherwise, where $\ci$ is a constant to be chosen later.

While $S$ is drawn from an uncountable set, we will only be concerned with events that depend upon a finite portion of it, and countable unions and intersections thereof. Therefore, we can use as our underlying $\sigma$-algebra that generated by the set of all events $E_{v,j} = \{S:S(v)_j = 1\}$, where $v,j \in \nat$, with the corresponding probabilities $\Prob{E_{v,j,1}} = \frac{\ci\log v}{j+2\ci\log v}$.

We set delay function $g(r,k) = \frac{\ci^2 r \log k}{\log\log k}$.

To simplify our task, we begin with some useful observations:

First we note that since we only care about the asymptotic behavior of $g$, we can assume that $r$ is larger than a sufficiently large constant.

We also note that we need not consider all $(r,k)$-instances. For a given $(r,k)$-instance and time-step $j$, let $K_j$ be the set of nodes woken up by time $j$ (with $k_j := |K_j|$), and $r_j$ be defined as $r$ but restricted to the nodes in $K_j$. Let $t$ be the earliest time-step such that $t>g(r_t,k_t)$, and curtail the $(r,k)$-instance to the corresponding $(r_t,k_t)$-instance of nodes in $K_t$. It is easy to see that if we hit all curtailed $(r_t,k_t)$-instances within $g(r_t,k_t)$ time, we must hit all $(r,k)$-instances within $g(r,k)$ time, so henceforth we will only consider curtailed instances (i.e., we can assume $j\leq g(r_j,k_j) \forall j< g(r,k)$).

Finally, we observe that, since nodes' behavior is not dependent on the global clock, we can shift all $(r,k)$-instances to begin at time-step $0$.

To analyze the probability of hitting $(r,k)$-instances, define the \emph{load} of a time-step $f(j)$ to be the expected number of transmissions in that time-step:

\begin{definition}
	For a fixed $(r,k)$-instance, the \textbf{load $f(j)$ of a time-step $j$} is defined as \[\sum_{v\in K_j} \Prob{v\text{ transmits}} = \sum_{v\in K_j} \frac{\ci \log v}{j-\omega(v)+2\ci \log v}\enspace.\]
\end{definition}

We now seek to bound the load from above and below, since when the load is close to constant we have a good chance of hitting.

\begin{lemma}\label{lem:loadlb}
	All time-steps $j\leq g(r,k)$ have $f(j) \geq \frac{\log\log k}{2\ci\log k}$.
\end{lemma}

\begin{proof}
	Fix a time-step $j\leq g(r,k)$, let $K_j$ be the set of nodes awake by time-step $j$, and let $k_j = |K_j|$ and $r_j=\sum_{v\in K_j} \log v$, analogous to $r$ and $k$. If $k_j = k$, then
	\[f(j) \geq \sum\limits_{v\in K}\frac{\ci\log v}{j+2\ci\log v} \geq \frac{\ci r}{j+2\ci r} \geq \frac{\ci r}{\frac{2\ci^2 r \log k}{\log\log k}} \geq \frac{\log\log k}{2\ci\log k}\enspace.\]
	If $k_j < k$, then due to our curtailing of instances, we have $j \leq g(r_j)$. So,
	\begin{align*}
	f(j) &\geq \sum\limits_{v\in K_j}\frac{\ci\log v}{j+2\ci \log v} \geq \frac{\ci r_j}{j+2\ci r_j} \geq \frac{\ci r_j}{\frac{2\ci^2 r_j \log k_j}{\log\log k_j}} \geq \frac{\log\log k_j}{2\ci\log k_j} \geq \frac{\log\log k}{2\ci\log k}\enspace.
	\qedhere
	\end{align*}
\end{proof}

Having bounded load from below, we also seek to bound it from above. Unfortunately, the load in any particular time-step can be very high, but we can get a good bound on at least a constant fraction of the columns.
\begin{lemma}
	Let $F$ denote the set of time-steps $j\leq g(r,k)$ such that $\frac{\log\log k}{2\ci\log k} \leq f(j) \leq \frac{\log\log k}{3}$. Then $|F|\geq \frac{\ci r\log k}{2\log\log k}$.
\end{lemma}

\begin{proof}
	The total load over all time-steps can be bounded as follows:
	
	\begin{align*}
	\sum_{j\leq g(r,k)} f(j) &= \sum_{j<g(r,k)} \sum\limits_{v\in K_j}\frac{\ci\log v}{j-\omega(v)+2\ci\log v}
	\, \leq\sum\limits_{v\in K} \sum_{\omega(v)< j<g(r,k)} \frac{\ci\log v}{j-\omega(v)+2\ci \log v}\\
	& \leq \sum\limits_{v\in K} \ci\log v \sum_{ j<g(r,k)} \frac{1}{j+2\ci \log v}
	\, \leq \sum\limits_{v\in K} \ci\log v \ln \frac{2g(r,k)}{4 \ci \log v}
	\enspace.
	\end{align*}
	
	Let $K_i=\{v\in K : \frac {r}{k\cdot 2^{i}} \leq \log v < \frac {r}{k\cdot 2^{i-1}}\}$, for $i\geq 1$, and $K' = \{v\in K : \log v \geq \frac {r}{k}\}$
	
	If $\sum_{v\in K_i} \log v > \frac {r}{2^i}$ then
	$
	r < 2^i \sum_{v\in K_i} \log v \leq 2^i \sum_{v\in K_i} \frac {r}{k\cdot 2^i}   \leq r
	\enspace.
	$
	This gives a contradiction, so we must have $\sum_{v\in K_i} \log v \leq \frac {r}{2^i}$. Then,
	
	\begin{align*}
	\sum_{j\leq g(r,k)} f(j)
	&\leq \sum\limits_{v\in K} \ci\log v \ln \frac{2g(r,k)}{4 \ci \log v}
	\leq\sum\limits_{i\geq 1} \sum\limits_{v\in K_i} \ci\log v \ln \frac{g(r,k)}{2 \ci \log v} +  \sum_{v\in K'} \ci\log v \ln \frac{g(r,k)}{2 \ci \log v}\\
	&\leq\sum\limits_{i\geq 1} \sum\limits_{v\in K_i} \ci\log v \ln \frac{g(r,k)}{2 \ci \frac {r}{k\cdot 2^{i}}} +  \sum_{v\in K'} \ci\log v \ln \frac{g(r,k)}{2 \ci \frac {r}{k}}\\
	&=\sum\limits_{i\geq 1} \sum\limits_{v\in K_i} \ci\log v \ln \frac{\ci k 2^{i-1}\log k}{\log\log k} +  \sum_{v\in K'} \ci\log v \ln \frac{\ci k\log k}{2\log\log k}\\
	&\leq \sum\limits_{i\geq 1}  \ci r 2^{-i} (2 \ln k + (i-1)\ln 2) +  2\ci r \ln k
	\leq 5 \ci r \ln k
	\leq 8\ci r \log k
	\enspace.
	\end{align*}
	
	Therefore, at most $\frac{24\ci r\log k}{\log\log k}$ time-steps have load higher than $\frac{\log\log k}{3}$. Since by Lemma \ref{lem:loadlb} all time-steps have load at least $\frac{\log\log k}{2\ci \log k}$,we have $|F|\geq g(r,k) - \frac{24\ci r\log k}{\log\log k} \geq \frac{\ci^2 r\log k}{2\log\log k}$ (provided we pick $\ci \geq 7$).
\end{proof}

Now that we have bounded load, we show how it relates to hitting probability. The following lemma, or variants thereof, has been used in several previous works such as \cite{-DM10}, but we prove it here for completeness.

\begin{lemma}\label{lem:colhit}
	Let $x_i$, $i\in [n]$, be independent $\{0,1\}$-valued random variables with $\Prob{x_i=1}\leq \frac 12$, and let $f =\sum_{i\in [n]} \Prob{x_i=1}$. Then $\Prob{\sum_{i\in [n]} x_i = 1} \geq f4^{-f}$.
\end{lemma}

\begin{proof}
	\begin{align*}
	\Prob{\sum_{i\in [n]} x_i = 1}
	&= \sum_{j\in [n]} \Prob{x_j = 1 \land x_i = 0 \forall i \neq j}
	\geq \sum_{j\in [n]} \Prob{x_j = 1}\cdot \Prob{x_i = 0 \forall i}\\
	&\geq f\cdot \Prob{x_i = 0 \forall i}
	= f\cdot \prod_{i\in [n]} (1- \Prob{x_i = 1})
	\geq f\cdot \prod_{i\in [n]} 4^{-\Prob{x_i = 1}} \\
	&= f\cdot 4^{-\sum_{i\in [n]} \Prob{x_i = 1}}
	=f4^{-f}
	\enspace.
	\qedhere
	\end{align*}
\end{proof}

We can use Lemma \ref{lem:colhit} to show that the probability that we hit on our `good' time-steps (those in $F$) is high:

\begin{lemma}\label{lem:goodcolhit}
	For any time-step $j \in F$, the probability that $j$ hits is at least $\frac{\log\log k}{3\ci\log k}$.
\end{lemma}

\begin{proof}
	$\frac{\log\log k}{ 2\ci\log k} \leq f(j) \leq \frac{\log\log k}{3}$, and so $f(j)4^{-f(j)}$ is minimized at $f(j) = \frac{\log\log k}{2\ci\log k}$ and is therefore at least $\frac{\log\log k}{2\ci\log k} 4^{-\frac{\log\log k}{2\ci\log k}}\geq \frac{\log\log k}{3\ci\log k}$.
\end{proof}

We now bound the number of possible instances we must hit:

\begin{lemma}\label{lem:risnt}
	For any (sufficiently large) $r$, the number of unique $(r,k)$-instances is at most $2^{5r}$.
\end{lemma}

\begin{proof}
	The total number of bits used in all node IDs in the instance is at most $r$. There are at most $2^{r+1}$ possible bit-strings of length at most $r$, and at most $2^r$ ways of dividing each of these into substrings (for individual IDs), giving at most $2^{2r+1}$ sets of node IDs. The number of possible wake-up functions $\omega:K\rightarrow \nat$ is at most $g(r,k)^k$, since all nodes must be awake by $g(r,k)$ time or the instance would have been curtailed.
	\begin{align*}
	g(r,k)^k &= 2^{k\log g(r,k)} \leq 2^{1.1 k\log r}
	= 2^{1.1 (k\log k+ k \log \frac rk)}
	\leq 2^{1.3 (k\log (k^{0.9}) + r)}
	\leq 2^{2.9 r}
	\enspace.
	\end{align*}
	
	So, the total number of possible $(r,k)$-instances is at most $2^{2r+1+2.9r} \leq 2^{5 r}$.
\end{proof}

\begin{lemma}\label{lem:hitins}
	For any (sufficiently large) $r$, the probability that $S$ does not hit all $(r,k)$-instances is at most $2^{-3r}$
\end{lemma}

\begin{proof}
	Fix some $(r,k)$-instance. The probability that it is not hit within $g(k,r)$ time-steps is at most
	\[\prod_{j\in F} (1- \frac{\log\log k}{3\ci\log k}) \leq e^{- |F|\frac{\log\log k}{3\ci\log k}} \leq e^{-\frac 23 \ci r} = 2^{-\frac {2 \ci r}{3\ln 2}}\enspace,\] by Lemma \ref{lem:goodcolhit}.
	Hence, if we set $c = 9$, by a union bound the probability that any $(r,k)$-instance is not hit is at most
	$2^{5 r} \cdot 2^{-\frac {18 r}{3\ln 2}}\leq 2^{-3r}\enspace.\qedhere$
\end{proof}

We can now prove our main theorem on unbounded universal synchronizers (Theorem \ref{thm:UURS}):

\begin{proof}
	By Lemma \ref{lem:hitins} and a union bound over $r$, the probability that $S$ does not hit all instances is at most $\sum_{r\in \nat} 2^{-3r} < 1$. Therefore $S$ is an unbounded universal synchronizer of delay $g$ with non-zero probability, so such an object must exist.
\end{proof}

\subsection{Unbounded transmission schedules}
\label{subsec-UTS}

For the task of broadcasting, i.e., when a global clock is available, we can make use of the global clock to improve our running time. We again define an infinite combinatorial object, which we will call an \textbf{unbounded transmission schedule}. We use the same definition of $(r,k)$-instances as in the previous section.

\begin{definition}
	\label{def:UTS}
	For a function $h:\nat \times \nat \rightarrow \nat$, an \textbf{unbounded transmission schedule of delay $h$} is a function $T : \nat \times \nat \rightarrow \{0,1\}^\nat$ such that $T(v,\omega(v))_j = 0$ for any $j < \omega(v)$, and for any $(r,k)$-instance there is some time-step $j \leq \omega(K) + h(r)$ with $\sum_{v\in K}  T(v,\omega(v))_{j} = 1$.
\end{definition}

The difference here from an unbounded universal synchronizer is that nodes now know the global time-step in which they wake up, and so their transmission patterns can depend upon it. This is the second argument of the function $T$. The other difference in the meaning of the definition is that the output of $T$ now corresponds to absolute time-step numbers, rather than being relative to each node's wake-up time as for unbounded universal synchronizers. That is, the $j^{th}$ entry of a node's output sequence tells it how it should behave in global time-step $j$, rather than $j$ time-steps after it wakes up.

Our existence result for unbounded transmission schedules is the following:

\begin{restatable}{theorem}{UTS}
	\label{thm:UTS}
	There exists an \textbf{unbounded transmission schedule of delay $h$}, where $h(r,k) =  O\left(r\log\log k\right)$.
\end{restatable}

Our method will again be to randomly generate a candidate unbounded transmission schedule $T$, and then prove that it satisfies the required property with positive probability, so some such object must exist.

Let $\cj$ be a constant to be chosen later. Our candidate object $T$ will be generated as follows: for each node $v$, we generate a transmission sequence $s_{v,j}$, $j \in \nat$, with $s_{v,j}$ independently chosen to be $1$ with probability $\frac{\cj\log v \log\log j}{j+2\cj\log v\log\log j}$ and $0$ otherwise.

However, these will not be our final probabilities for generating $T$. To make use of our global clock, we will also divide time into short \emph{phases} during which transmission probability will decay exponentially. The lengths of these phases will be based on a function $z(j) := 2^{\lceil 1+ \log\log\log j\rceil}$, i.e., $\log\log j$ rounded up to the next-plus-one power of $2$. An important property to note is that for all $i$, $z(i)|z(i+1)$. We also generate a sequence $p_{v,j}$, $j \in \nat$ of \emph{phase values}, each chosen independently and uniformly at random from the real interval $[0,1]$. These, combined with the global time-step number and current phase length, will give us our final generation probabilities.

We set $T(v,\omega(v))_j$ to equal $1$ iff $s_{v,j-\omega(v)}=1$ and $p_{v,j-\omega(v)}\leq 2^{-j\bmod z(j-\omega(v))}$.

It can then be seen that \[\Prob{T(v,\omega(v))_j = 1} = \frac{\cj\log v \log\log (j-\omega(v))}{(j-\omega(v)+2\cj\log v\log\log (j-\omega(v))) 2^{j\bmod z(j-\omega(v))}}\enspace .\]

The reason we split the process of random generation into two steps (using our phase values) is that now, if we shift all wake-up times in an $(r,k)$-instance by the same multiple of $z(x)$, then node behavior in the first $x$ time-steps after $\omega(K)$ does not change. We will require this property when analyzing $T$.

Our probabilistic analysis is with respect to the $\sigma$-algebra generated by all events $E_{v,\omega(v),j}= \{T:T(v,\omega(v))_j = 1\}$, with $v,\omega(v),j\in \nat$, and with the corresponding probabilities given above.

Let \textbf{load $f(j)$ of a time-step $j$} be again defined as the expected number of transmissions in that time-step:

\[
f(j) :=
\sum_{v \in K_j}\frac{\cj\log v \log\log (j-\omega(v))}{(j-\omega(v)+2\cj\log v\log\log (j-\omega(v))) 2^{j\bmod z(j-\omega(v))}}
\enspace.
\]

We will set our delay function $h(r,k) = \cj^2 r\log\log k$.

Again we make some observations that allow us to narrow the circumstances we must consider: firstly that we can again assume that $r$ is larger than a sufficiently large constant, and secondly that we need only look at curtailed instances (i.e., we can assume $j-\omega(K)\leq h(r_j,k_j) \forall j< h(r,k)$). This time, however, we cannot shift instances to begin at time-step $0$, because node behavior is dependent upon global time-step number.

We follow a similar line of proof as before, except with some extra complications in dealing with phases. We first consider only time-steps at the beginning of each phase, i.e., time-steps $\omega(K)<j\leq \omega(K)+h(r,k)$ with $j\bmod z(h(r,k))\equiv 0$, and we will call these \emph{basic} time-steps. Notice that for basic time-steps,
\[f(j) = \sum_{v \in K_j}\frac{\cj\log v \log\log (j-\omega(v))}{j-\omega(v)2\cj\log v\log\log(j-\omega(v))}\enspace.\]

We bound the load of basic time-steps from below:

\begin{lemma}\label{lem:loadlb}
	All basic time-steps $j$ have $f(j) \geq \frac {1}{2\cj}$.
\end{lemma}

\begin{proof}
	Fix a basic time-step $j$, let $K_j$ be the set of nodes awake by time-step $j$, and let $k_j = |K_j|$ and $r_j=\sum_{v\in K_j} \log v$, analogous to $r$ and $k$. If $k_j = k$, then
	\begin{align*}
	f(j) &\geq \sum\limits_{v\in K}\frac{\cj\log v \log\log (j-\omega(v))}{j-\omega(v)+2\cj\log v\log\log(j-\omega(v))}
	\geq \sum\limits_{v\in K}\frac{\cj\log v \log\log h(r,k)}{h(r,k)+2\cj\log v\log\log h(r,k)}\\
	&\geq \sum\limits_{v\in K}\frac{\cj\log v \log\log k}{2\cj^2 r\log\log k}
	\geq \frac{\cj r}{2\cj^2 r} = \frac{1}{2\cj}\enspace.
	\end{align*}
	
	If $k_j < k$, then due to our curtailing of instances, we have $j-\omega(K) \leq h(r_j,k_j)$. So,
	\begin{align*}
	f(j) &\geq \sum\limits_{v\in K_j}\frac{\cj\log v \log\log (j-\omega(v))}{j-\omega(v)+2\cj\log v\log\log(j-\omega(v))}
	\geq \sum\limits_{v\in K}\frac{\cj\log v \log\log h(r_j,k_j)}{h(r_j,k_j)+2\cj\log v\log\log h(r,k)}\\
	&\geq \sum\limits_{v\in K}\frac{\cj\log v \log\log k_j}{2\cj^2 r_j\log\log k_j}
	\geq \frac{\cj r_j}{2\cj^2 r_j}
	= \frac{1}{2\cj}
	\enspace.
	\qedhere
	\end{align*}
\end{proof}

We next examine time-steps at the end of phases, i.e., with $\omega(K)<j\leq \omega(K)+h(r,k)$ and $j\bmod z(h(r,k))\equiv -1$, and call these \emph{ending} time-steps. Note that for ending time-steps, \[f(j) = \sum_{v \in K_j}\frac{\cj\log v \log\log (j-\omega(v))}{(j-\omega(v)+2\cj\log v\log\log(j-\omega(v))) 2^{z(j-\omega(v))-1}}\enspace.\]

We bound the load of (a constant fraction of) ending time-steps from above:
\begin{lemma}\label{lem:a}
	Let $\mathcal F$ denote the set of ending time-steps $j$ such that $ f(j) \leq 1$. Then $|\mathcal F|\geq \frac{\cj^2r}{2}$.
\end{lemma}

\begin{proof}
	Let $t$ be the first ending time-step.
	The total load over all ending time-steps can be bounded as follows:
	
	\[
	\sum_{\text{ending timestep }j} f(j) \leq \sum_{i = 0}^{h(r,k)/z(h(r,k))} f(t+iz(h(r,k)))\\
	\leq \sum_{i = 0}^{\cj^2 r} f(t+iz(h(r,k)))\enspace.
	\]
	
	Applying the definition of $f$, $f(t+iz(h(r,k)))$ is equal to:
	\[ \sum_{v \in K_{t+iz(h(r))}}\frac{\cj\log v \log\log (t+iz(h(r,k))-\omega(v))}{(t+iz(h(r,k))-\omega(v)+2\cj\log v\log\log(t+iz(h(r,k))-\omega(v))) 2^{z(t+iz(h(r,k))-\omega(v))-1}}
	\enspace,\]
	
	which is bounded from above when $t-\omega(v) = 0$:
	
	\begin{align*}
	f(t+iz(h(r,k))) &\leq \sum_{v \in K_{t+iz(h(r))}}\frac{\cj\log v \log\log (iz(h(r)))}{(iz(h(r,k)) + 2\cj \log v\log\log(iz(h(r,k))))2^{z(iz(h(r,k)))}}\\
	&\leq \sum_{v \in K_{t+iz(h(r,k))}}\frac{\cj\log v \log\log (iz(h(r,k)))}{iz(h(r,k)) 2^{z(iz(h(r,k)))}}
	\enspace.
	\end{align*}
	
	So,
	\begin{align*}
	\sum_{\text{ending timestep }j} f(j) &\leq \sum_{i = 0}^{\cj^2 r} \sum_{v \in K_{t+iz(h(r,k))}}\frac{\cj\log v \log\log (iz(h(r,k)))}{iz(h(r,k))2^{z(iz(h(r,k)))}}\\
	&\leq \sum_{v \in K}\sum_{i = 0}^{\cj^2 r} \frac{\cj\log v \log\log (iz(h(r,k)))}{iz(h(r,k))2^{z(iz(h(r,k)))}}
	\leq \sum_{v \in K}\sum_{i = 0}^{\cj^2 r} \frac{2\cj \log v \log\log (iz(h(r,k)))}{iz(h(r,k))\log^2 (iz(h(r,k)))}\\
	&\leq \sum_{v \in K}2 \cj \log v \sum_{i = 0}^{\infty} \frac{\log\log (iz(h(r,k))) }{iz(h(r,k))\log^2 (iz(h(r,k)))}
	\leq 10\cj r\enspace.
	\end{align*}
	Here the last inequality follows since the second sum converges to a constant less than $5$, which can be seen by inspection of the first three terms and using the integral bound $\int_{2}^{\infty} \frac{\log\log x}{x\log^2 x}<2$
	
	Therefore, at most $10\cj r$ ending time-steps have load higher than $1$, and so at least $\cj^2 r - 10\cj r \geq \frac{\cj^2r}{2} $ (provided we set $\cj\geq 5$) ending time-steps have $f(j) \leq 1$.
\end{proof}

We can use Lemma \ref{lem:a} to show that we have sufficiently many time-steps with load within the interval $[\frac 1\cj,1]$:

\begin{lemma}
	Let $\mathcal F$ be the set of time-steps $\omega(K)<j\leq \omega(K)+h(r,k)$ with $\frac {1}{2\cj} \leq f(j) \leq  1$. Then $|\mathcal F|\geq \frac{\cj^2r}{2}$.
\end{lemma}

\begin{proof}
	It can be seen that load decreases by at most a multiplicative factor of $3$ between consecutive time-steps (since the contribution of each node decreases by at most a factor of $3$). So, since every base time-step has load at least $\frac {1}{2\cj}$, for every ending timestep $j$ with $f(j) \leq  1$, there is some $j'$ in the preceding phase with $\frac {1}{2\cj} \leq f(j') \leq 1$.
\end{proof}

Since these time-steps have constant load, they have constant probability of hitting:

\begin{lemma}\label{lem:goodcolhit2}
	For any time-step $j \in \mathcal F$, the probability that $j$ hits is at least $\frac {1}{3\cj}$.
\end{lemma}

\begin{proof}
	By Lemma \ref{lem:colhit}, the probability that $j$ hits is at least $f(j)4^{-f(j)}$. This is minimized over the range $[\frac {1}{2\cj},1]$ at $f(j) = \frac {1}{2\cj}$ and is therefore at least $\frac{4^{-\frac{1}{2\cj}}}{2\cj} \geq \frac {1}{3\cj}$.
\end{proof}

We now need to know how many unique $(r,k)$-instances we must hit. Since we are only concerned with the first $h(r,k)$ time-steps after the first node wakes up, we need only consider $(r,k)$-instances which are unique within this time range.

\begin{lemma}
	For any (sufficiently large) $r$, the number of unique (up to the first $h(r,k)$ steps after activation) $(r,k)$-instances is at most $2^{5r}$.
\end{lemma}

\begin{proof}
	As before (in Lemma \ref{lem:risnt}) there are at most $2^{2r+1}$ sets of node IDs. The number of possible wake-up functions $\omega:Q\rightarrow \nat$ for a fixed $\omega(K)$ is again at most $h(r,k)^q$, and though we are using a different delay function to the previous section, the calculations used to prove Lemma \ref{lem:risnt} still hold.
	
	\begin{align*}
	h(r)^k &= 2^{k\log h(r,k)} \leq 2^{1.1 k\log r}
	=
	2^{1.1 (k\log k + k \log \frac rk)}
	\leq
	2^{1.3 (k\log (k^{0.9}) + r)}
	\leq
	2^{2.9 r}
	\enspace.
	\end{align*}
	
	However, now our object does depend on $\omega(K)$, though as we noted we can shift all activation times by a multiple of $z(h(r,k))$ and behavior during the time-steps we analyze is unchanged. So our total number of instances to consider is multiplied by $z(h(r,k))$, and is upper bounded by
	$2^{2r+1+2.9r} z(h(r,k)) \leq 2^{5 r}
	\enspace.\qedhere$

\end{proof}

\begin{lemma}\label{lem:hitins2}
	For any (sufficiently large) $r$, the probability that $S$ does not hit all $(r,k)$-instances is at most $2^{-3r}$.
\end{lemma}

\begin{proof}
	Fix some $(r,k)$-instance. The probability that it is not hit within $h(r,k)$ time-steps is at most
	\[\prod_{j\leq |\mathcal F|} (1-\frac {1}{3\cj}) \leq e^{-\frac {|\mathcal F|}{3\cj}} \leq e^{-\frac {\cj r}{6}} = 2^{-\frac {\cj r}{6\ln 2}} \enspace.\]
	
	Hence, if we set $\cj = 34$, by a union bound the probability that any $(r,k)$-instance is not hit is at most
	$2^{5 r} \cdot 2^{-\frac {34 r}{6\ln 2}} \leq 2^{-3r}\enspace.\qedhere$
\end{proof}

We can now prove our main theorem on unbounded transmission schedules (Theorem \ref{thm:UTS}):

\begin{proof}
	By Lemma \ref{lem:hitins2} and a union bound over $r$, the probability that $S$ does not hit all instances is at most $\sum_{r\in \nat} 2^{-3r} < 1$. Therefore $T$ is an unbounded transmission schedule of delay $h$ with non-zero probability, so such an object must exist.
\end{proof}

\section{Algorithms for multiple access channels}
\label{sec:algs-multiple-access-channels}

Armed with our combinatorial objects, our algorithms are now extremely simple, and are the same for multiple access channels as for multi-hop radio networks.

Let $S$ be an unbounded universal synchronizer of delay $g$, where $g(r,k) =  O\left(\frac{r \log k}{\log\log k}\right)$, and $T$ be an unbounded transmission schedule of delay $h$, where $h(r,k) = O(r \log\log k)$.

Our algorithms are simply applications of these objects.

\begin{algorithm}[H]
	\caption{Wake-up at a node $v$}
	\label{alg:WU}
	\begin{algorithmic}
		\For {$j$ from $1$ to $\infty$, in time-step $\omega(v)+j$,}
		\State $v$ transmits iff $S(v)_{j} = 1$
		\EndFor	
	\end{algorithmic}
\end{algorithm}	

\begin{theorem}\label{thm:MACWU}
	Algorithm \ref{alg:WU} performs wake-up in multiple access channels in time $O\left(\frac{k \log L \log k}{\log\log k}\right)$, without knowledge of $k$ or $L$.
\end{theorem}

\begin{proof}
	By the definition of an unbounded universal synchronizer, there is some time-step within
	\[g(r,k) = O\left(\frac{r \log k}{\log\log k}\right) = O\left(\frac{k \log L \log k }{\log\log k}\right)\]
	time-steps of the first activation in which only one node transmits, and this completes wake-up.
\end{proof}

\begin{algorithm}[H]
	\caption{Broadcasting at a node $v$}
	\label{alg:BC}
	\begin{algorithmic}
		\For {$j$ from $1$ to $\infty$, in time-step $j$,}
		\State $v$ transmits iff $T(v,\omega(v))_{j} = 1$
		\EndFor	
	\end{algorithmic}
\end{algorithm}	

\begin{theorem}\label{thm:MACBC}
	Algorithm \ref{alg:BC} performs broadcasting in multiple access channels in time \sloppy{$O(k \log L \log\log k)$}, without knowledge of $k$ or $L$.
\end{theorem}

\begin{proof}
	By the definition of an unbounded transmission schedule, there is some time-step within
	$h(r,k) = O(r \log\log k) = O(k \log L \log\log k)$
	time-steps of the first activation in which only one node transmits, and this completes broadcasting.
\end{proof}

\section{Algorithms for radio networks}
\label{sec:algs-radio-networks}

To see how our results on multiple access channels (Theorems \ref{thm:MACWU} and \ref{thm:MACBC}) transfer directly to multi-hop radio networks, we apply the analysis method of \cite{-CD16} for radio network protocols. The idea is that we fix a shortest path $p=(p_0,p_1,\dots p_d)$ from some \emph{source} node to some \emph{target} node, and then classify all nodes into \emph{layers} depending on the furthest node along the path to which they are an in-neighbor, i.e., layer $L_i = \{v:\max j \text{ such that } (v,p_j)\in E = i \}$. We wish to quantify how long a layer can remain \emph{leading}, that is, the furthest layer to contain awake nodes. The key point is that we can regard these layers as multiple access channels: though they are not necessarily cliques, all we need is for a single node in the layer to transmit and then the layer ceases to be leading. Once the final layer ceases to be leading, the target node must be informed, and since this node was chosen arbitrarily we obtain a time-bound for informing the entire network. Thereby the problem is reduced to a sequence of at most $D$ single-hop instances, whose sizes sum to at most $n$. For full details of this analysis method see \cite{-CD16}.

Therefore we can use the following lemma from \cite{-CD16} (paraphrased to fit our notation) to analyze how our algorithms perform in radio networks.

\begin{lemma}{\textbf{(Lemma 10 of \cite{-CD16})}}\label{lem:algtime}
	Let $X:[n]\rightarrow \nat$ be a non-decreasing function, and define $Y(n)$ to be the supremum of the function $\sum_{i=1}^{n} X(\ell_i)$, where non-negative integers $\ell_i$ satisfy the constraint $\sum_{i=1}^{n} \ell_i \le n$. If a broadcast or wake-up protocol ensures that any layer of size $\ell$ remains leading for no more than $X(\ell)$ time-steps, then all nodes wake up within $Y(n)$ time-steps.
\end{lemma}

\begin{theorem}
	\label{thm:RNWU}
	Algorithm \ref{alg:WU} performs wake-up in radio networks in time $O(\frac{n \log L \log n}{\log\log n})$, without knowledge of $n$ or $L$.
\end{theorem}

\begin{proof}
	By Theorem \ref{thm:MACWU}, any layer of size $\ell$ remains leading for no more than $X(\ell)$ time-steps, where $X(\ell) = O(\frac{\ell \log L \log \ell}{\log\log \ell})$. $Y(n,h)$ is then maximized by setting $\ell_1 = n$ and $\ell_i = 0$ for every $i>1$. So, by Lemma \ref{lem:algtime}, wake-up is performed for the entire radio network in $O(\frac{n \log L \log n}{\log\log n})$ time.
\end{proof}

\begin{theorem}
	\label{thm:RNB}
	Algorithm \ref{alg:BC} performs broadcasting in radio networks in time $O(n \log L \log\log n)$, without knowledge of $n$ or $L$.
\end{theorem}

\begin{proof}
	By Theorem \ref{thm:MACBC}, any layer of size $\ell$ remains leading for no more than $X(\ell)$ time-steps, where $X(\ell) = O(\ell \log L \log\log \ell)$. $Y(n,h)$ is then maximized by setting $\ell_1 = n$ and $\ell_i = 0$ for $i>1$. So, by Lemma \ref{lem:algtime}, broadcasting is performed for the entire radio network in $O(n \log L \log\log n)$ time.
\end{proof}

\section{Conclusions}

We have shown that \emph{deterministic algorithms} for communications primitives in multiple access channels and multi-hop radio networks \emph{need not assume parameter knowledge, or small IDs}, to be efficient. One possible next step would be to show a means by which nodes could generate efficient transmission schedules based on some finite (i.e. with size bounded by some function of ID) advice string; this  would go some way towards making the algorithm practical. 

\newcommand{\Proc}{Proceedings of the\xspace}
\newcommand{\STOC}{Annual ACM Symposium on Theory of Computing (STOC)}
\newcommand{\FOCS}{IEEE Symposium on Foundations of Computer Science (FOCS)}
\newcommand{\SODA}{Annual ACM-SIAM Symposium on Discrete Algorithms (SODA)}
\newcommand{\COCOON}{Annual International Computing Combinatorics Conference (COCOON)}
\newcommand{\DISC}{International Symposium on Distributed Computing (DISC)}
\newcommand{\ESA}{Annual European Symposium on Algorithms (ESA)}
\newcommand{\ICALP}{Annual International Colloquium on Automata, Languages and Programming (ICALP)}
\newcommand{\IPL}{Information Processing Letters}
\newcommand{\JACM}{Journal of the ACM}
\newcommand{\JALGORITHMS}{Journal of Algorithms}
\newcommand{\JCSS}{Journal of Computer and System Sciences}
\newcommand{\PODC}{Annual ACM Symposium on Principles of Distributed Computing (PODC)}
\newcommand{\SICOMP}{SIAM Journal on Computing}
\newcommand{\SPAA}{Annual ACM Symposium on Parallelism in Algorithms and Architectures (SPAA)}
\newcommand{\STACS}{Annual Symposium on Theoretical Aspects of Computer Science (STACS)}
\newcommand{\TALG}{ACM Transactions on Algorithms}
\newcommand{\TCS}{Theoretical Computer Science}



\bibliography{nk}

\end{document}